\documentclass[12pt,english,a4paper,oneside]{article}
\usepackage[T1]{fontenc}
\usepackage[latin9]{inputenc}
\usepackage{geometry}
\usepackage{float}
\usepackage{amsthm}
\usepackage{amsmath}
\usepackage{amssymb}
\usepackage{graphicx}
\usepackage{setspace}
\makeatletter

\numberwithin{equation}{section}
\numberwithin{figure}{section}
\usepackage[color=green!40]{todonotes}
\theoremstyle{plain}
\newtheorem{thm}{\protect\theoremname}[section]
\theoremstyle{plain}
\newtheorem{prop}[thm]{\protect\propositionname}
\theoremstyle{plain}
\newtheorem{fact}[thm]{\protect\factname}
\newcounter{casectr}
\newenvironment{caseenv}
{\begin{list}{{\itshape\ \protect\casename} \arabic{casectr}.}{%
            \setlength{\leftmargin}{\labelwidth}
                \addtolength{\leftmargin}{\parskip}
                    \setlength{\itemindent}{\listparindent}
                    \setlength{\itemsep}{\medskipamount}
                    \setlength{\topsep}{\itemsep}}
                \setcounter{casectr}{0}
            \usecounter{casectr}}
{\end{list}}
\theoremstyle{plain}

\theoremstyle{plain}

\theoremstyle{plain}

\theoremstyle{remark}

\@ifundefined{date}{}{\date{}}
\usepackage[all]{xy}
\let\oldenumerate=\enumerate
\def\enumerate{\oldenumerate%
    \setlength{\itemsep}{3pt}\setlength{\parsep}{3pt}}%
\usepackage{mathrsfs}
\usepackage{makeidx}
\makeindex
\usepackage[english]{babel}
\addto\captionsenglish{%
    
}
\usepackage[rule=0.5pt,totoc]{idxlayout}
\usepackage[linktoc=all]{hyperref}
\hypersetup{
    colorlinks,
    citecolor=black,
    filecolor=black,
    linkcolor=black,
    urlcolor=black
}
\usepackage[margin=10pt,font=small,labelfont=bf,labelsep=endash]{caption}
\usepackage{amssymb}
\usepackage{tocloft}
\usepackage{mparhack}
\usepackage{algpseudocode}
\usepackage{cleveref}
\usepackage{fancyhdr}

\usepackage{pdfpages}

\makeatother

\usepackage{babel}
\providecommand{\algorithmname}{Algorithm}
\providecommand{\casename}{Case}
\providecommand{\claimname}{Claim}
\providecommand{\corollaryname}{Corollary}
\providecommand{\factname}{Fact}
\providecommand{\lemmaname}{Lemma}
\providecommand{\propositionname}{Proposition}
\providecommand{\theoremname}{Theorem}

\begin{document}
    \global\long\def\ang#1{\langle#1\rangle}
    \global\long\def\mod#1#2{\equiv#1\ (mod\ #2)}
    \global\long\def\poly{a_{n}x^{n}+\ldots +a_{1}x+a_{0}}
    \global\long\def\polyr{\leq_{p}}
    \global\long\def\mapr{\leq_{m}}
    \global\long\def\ol#1{\overline{#1}}
    \global\long\def\u{\cup}
    \global\long\def\U#1#2{{\displaystyle \bigcup_{#1}^{#2}}}
    \global\long\def\i{\cap}
    \global\long\def\I#1#2{{\displaystyle \bigcap_{#1}^{#2}}}
    \global\long\def\P#1#2{{\displaystyle \prod_{#1}^{#2}}}
    \global\long\def\and{\wedge}
    \global\long\def\map#1#2#3{#1:#2\rightarrow#3}
    \global\long\def\normal{\trianglelefteq}
    \global\long\def\reals{\mathbb{R}}
    \global\long\def\ints{\mathbb{Z}}
    \global\long\def\nats{\mathbb{N}}
    \global\long\def\complex{\mathbb{C}}
    \global\long\def\dy#1#2{\frac{d#1}{d#2}}
    \global\long\def\goes{\rightarrow}
    \global\long\def\bin#1{\{0,1\}^{#1}}
    \global\long\def\pofn#1{p_{#1}(n)}
    \global\long\def\rofn#1{r_{#1}(n)}
    \global\long\def\maps#1#2{:#1\rightarrow#2}
    \global\long\def\adj{\sim}
    \global\long\def\nadj{\not\sim}
    \global\long\def\floor#1{\left\lfloor #1\right\rfloor}
    \global\long\def\ceil#1{\left\lceil #1\right\rceil }
    \global\long\def\geq{\geqslant}
    \global\long\def\leq{\leqslant}
    \global\long\def\matclass{\mathcal{M}}
    \global\long\def\graphclass{\mathcal{G}}

    \renewcommand{\thefootnote}{\fnsymbol{footnote}}

    \title{Matrix Partitions of Split Graphs}
    \author{Tom\'{a}s Feder
            \footnote{268 Waverley St.,
                      Palo Alto, CA 94301, USA;
                      tomas@theory.stanford.edu},
            Pavol Hell
            \footnote{School of Computing Science,
                      Simon Fraser University,
                      Burnaby, B.C., Canada, V5A 1S6},
            Oren Shklarsky
            \footnote{School of Computing Science,
                      Simon Fraser University,
                      Burnaby, B.C., Canada, V5A 1S6}}
    \maketitle
    
    \begin{abstract} 
        Matrix partition problems generalize a number of natural
        graph partition problems, and have been studied for several standard
        graph classes. We prove that each matrix partition problem has only
        finitely many minimal obstructions for split graphs. Previously such a
        result was only known for the class of cographs. (In particular, there
        are matrix partition problems which have infinitely many minimal chordal
        obstructions.)  We provide (close) upper and lower bounds on the maximum
        size of a minimal split obstruction. This shows for the first time that
        some matrices have exponential-sized minimal obstructions of any kind
        (not necessarily split graphs). We also discuss matrix partitions for
        bipartite and co-bipartite graphs.
    
         \end{abstract}

    \section{\label{sec:Introduction}Introduction}

        The approach to graph partition problems, proposed in
        \cite{Hell2012,Feder2006Perfect,Feder2003}, and used in this paper, is
        informed by the following distinction between different partition
        problems.

        There are graph partition problems which may be solved in polynomial
        time and for which the set of minimal non-partitionable graphs is
        finite. The \emph{split graphs recognition problem} is a well-known
        example \cite{Foeldes1977}. On the other hand there are partition
        problems, such as the \emph{bipartition} problem, which may be solved in
        polynomial time \cite{KonigBipartite}, but for which the set of minimal
        non-partitionable graphs is infinite (in the case of the bipartition
        problem, these are the odd cycles). Finally, there are numerous
        $NP$-complete graph partition problems, such as the \emph{3-colouring}
        problem.

        When discussing classes of partition problems, we will use \emph{patterns}
        to describe the requirements of a partition.  In particular, the
        patterns we examine specify partition problems in which the input
        graph's vertices are to be partitioned into independent sets, or
        cliques, or some combination of independent sets and cliques.  Further,
        we might require that two parts of vertices in the partition be
        completely adjacent, or completely non-adjacent.  Formally, we use
        \emph{matrices} to describe these patterns.

        Let $M$ be a symmetric $m\times m$ matrix over ${0,1,*}$. An
        $M$\emph{-partition} of a graph $G$ is a partition of the vertices of
        $G$ into parts $P_{1},P_{2},\ldots ,P_{m}$ such that two distinct
        vertices in parts $P_{i}$ and $P_{j}$ (possibly with $i=j$) are
        adjacent if $M(i,j)=1$, and nonadjacent if $M(i,j)=0$. The entry
        $M(i,j)=*$ signifies no restriction.

        Note that when $i=j$ these restrictions mean that part $P_{i}$ is either a
        clique, or an independent set, or is unrestricted, when $M(i,i)$ is $1$, or
        $0$, or $*$, respectively. Further, some of the parts may be empty. We may
        therefore assume that non of the diagonal entries of $M$ are asterisks or
        else the problem is trivial. For a fixed matrix $M$, the $M$-partition
        problem asks whether or not an input graph $G$ admits an $M$-partition.

        If a graph $G$ fails to admit an $M$-partition, we say that $G$ is an
        $M$-\emph{obstruction.} Further, if $G$ is an $M$-obstruction but
        deleting \emph{any} vertex of $G$ results in an $M$-partitionable graph,
        then $G$ is a \emph{minimal} $M$-obstruction.

        Given a graph $G$ and lists $L(v)\subseteq\{1,\ldots ,m\}$, with $v\in
        V(G)$, the \emph{list $M$-partition} problem asks whether $G$ admits an
        $M$-partition respecting the lists. That is, an $M$-partition of $G$ such
        that, for every $v\in V(G)$, the vertex $v$ is placed in a part $P_{i}$ only
        if $i\in P_{i}$. Note that diagonal asterisks do not make the problem
        trivial when lists are involved.  In this paper, we will focus on the
        non-list version, and will explicitly refer to the list version when it is
        discussed.

        For any matrix $M$ in this paper, we assume that there are $k$ zero entries
        and $\ell$ one entries on $M$'s diagonal. By row and column permutations, we
        may further assume that $M(0,0)=M(1,1)=\ldots =M(k,k)=0$ and
        $M(k+1,k+1)=\ldots =M(k+\ell,k+\ell)=1$. Let $A$ be the submatrix of $M$ on
        rows $1,\ldots ,k$ and columns $1,\ldots ,k$; let $B$ be the submatrix of
        $M$ on rows $k+1,\ldots ,m$ and columns $k+1,\ldots ,k$; and let $C$ be the
        submatrix of $M$ on rows $1,\ldots ,k$ and columns $k+1,\ldots ,m$.  When
        $M$ has no diagonal asterisks, $k+\ell=m$, and we say that $M$ is in
        $(A,B,C)$\emph{-block} form.

        Feder et al.\ have shown that if there are asterisks in block $A$ or
        block $B$ of a matrix $M$, then there are infinitely many minimal
        $M$-obstructions \cite{Feder06matrixpartitions}. Thus, when discussing
        general graphs, we must restrict our attention to matrices in which the
        only asterisk entries (if any) are in the block $C$. Such matrices are
        called \emph{friendly}. Of these, for any $m\times m$ matrix $M$
        containing no asterisk entries at all (i.e. having only entries in
        $\{0,1\}$), it has been shown that the largest minimal $M$-obstruction
        is of size $(k+1)(\ell+1)$ \cite{Feder2008}.

        Even when restricted to chordal graphs, there are matrices for which there
        are infinitely many chordal minimal obstructions
        \cite{Feder2005,FederHellRizi}.  One of these matrices and an infinite
        family of chordal minimal obstructions to this matrix, appear frequently in
        relation to other classes of graphs in this paper, and so are listed in
        Figure~\ref{fig:M_3_1_and_obs}.  The obstruction family in this figure is
        in fact a family of interval graphs, so that the matrix has
        infinitely many interval minimal obstructions.  Nonetheless, for any matrix
        $M$, the $M$-partition problem restricted to interval graphs can be solved
        in polynomial time \cite{Valadkhan2013}.  Note that the family in
        Figure~\ref{fig:M_3_1_and_obs} is not a family of split graphs, as each
        member contains $2K_{2}$ as an induced subgraph.

        \begin{figure}
            \centering{%
                \includegraphics{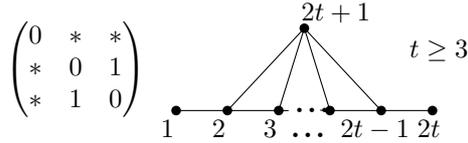}}
            \caption{\label{fig:M_3_1_and_obs} A matrix with a family of %
                                        infinitely many minimal obstructions.}
        \end{figure}

        For general matrices M, all known upper bounds on the size of minimal
        obstructions to M-partition are exponential
        \cite{Feder2008,Feder2006Cographs,Hell2012}; however, in none of these
        cases has it been shown that exponential-sized minimal obstructions to
        M-partition actually exist.

        This paper is organized as follows: In
        Section~\ref{sec:Matrix-Partitions-of-split-graphs}, we show that for
        any $m\times m$ matrix $M$, a split minimal $M$-obstruction has at most
        $O(m^2\cdot2^m)$ vertices. This implies that any $M$-partition problem
        (without lists) is solvable in polynomial time when the input is
        restricted to split-graphs.

        Section~\ref{sec:A-Special-Class} exhibits, for a particular class of
        $m\times m$ matrices, a split minimal obstruction of size
        $\Omega(2^{m})$, demonstrating that the exponential upper bound
        derived in Section~\ref{sec:Matrix-Partitions-of-split-graphs} is nearly
        tight. As noted above, this means that the class of split graph
        obstructions is the first class with finite minimal obstructions known 
        to contain exponentially large obstructions.

        In section~\ref{sec:Generalized-Split-Graphs}, we discuss graphs that admit
        other types of partitions, such as bipartite graphs and co-bipartite graphs.
        It is shown that for these classes also there are only finitely many minimal
        obstructions for any matrix $M$. These graph classes (including the class of
        split graphs) have a natural common generalization, namely graphs whose
        vertex set may be partitioned into $k$ independent sets and $\ell$ cliques,
        sometimes called $(k,\ell)$-graphs. Split graphs are $(1,1)$-graphs,
        bipartite graphs are $(2,0)$-graphs, and co-bipartite graphs are
        $(0,2)$-graphs. By contrast we show that when $k+\ell\geq3$, there is a
        matrix $M$ with infinitely many minimal $(k,\ell)$-graph obstructions. When
        $k\geq2$, there are infinitely many minimal $(k,\ell)$-graph obstructions
        that are chordal.

    \section{\label{sec:Matrix-Partitions-of-split-graphs}Matrix Partitions of
    Split Graphs}

        In this section we prove the following theorem.
        \begin{thm} \label{thm:split-graph-finitely-many-obs}
            If $M$ is a matrix with
            no diagonal asterisks, and $k\geq\ell$, then there are finitely many
            split minimal $M$-obstructions.
        \end{thm}

        A set of vertices $H\subseteq V(G)$ is said to be \emph{homogeneous} in $G$
        if the vertices of $V(G)-H$ can be partitioned into two sets, $S_{1}$ and
        $S_{2}$ such that every vertex of $S_{1}$ is adjacent to every vertex of
        $H$, and no vertex of $S_{2}$ is adjacent to a vertex of $H$. The proof of
        Theorem \ref{thm:split-graph-finitely-many-obs} relies on the existence of
        large homogenous sets in $M$-partitionable split graphs.

        \begin{prop} \label{prop:large-hom-set-split-graph}
            Let $A$ be a $k\times k$ matrix whose diagonal entries are all zero. Let
            $G_{A}$ be a split graph that admits an $A$-partition. Then every part
            $P$ of an $A$-partition of $G_{A}$ contains a homogeneous set in $G_{A}$
            of size at least $\frac{|P|-1}{2^{k-1}}$.
        \end{prop}
        \begin{proof}
            Suppose the parts of the $A$-partition of $G_{A}$ are $P_{1},...P_{k}$.
            Let $C\u I$ be a partition of $V(G_{A})$ into a clique $C$ and
            independent set $I$. Note that for $1\leq i\leq k$, we have that
            $|P_{i}\i C|\leq1$, since each $P_{i}$ is an independent set. Now, the
            vertices in the set $P_{1}\i I$ are non-adjacent to all but at most
            $k-1$ vertices, one in each $P_{i}\i C$, for $2\leq i\leq k$ (see Figure
            \ref{fig:bipartite_split_graph}). Assume without loss of generality that
            $|P_{i}\i C|=1$ and let $u_{i}\in P_{i}\i C$, for $2\leq i\leq k$. As
            each $u_{i}$ is either adjacent to at least half of the vertices of
            $P_{1}\i I$, or non-adjacent to at least half of the vertices of
            $P_{1}\i I$, a homogeneous set of size at least
            $\frac{|P_{1}|-1}{2^{k-1}}$ can be found in $|P_{1}|$. Since this
            argument may be repeated for any other part in the partition, we have
            the desired conclusion.

            \begin{figure}[H]
            \noindent \begin{centering}
            \includegraphics{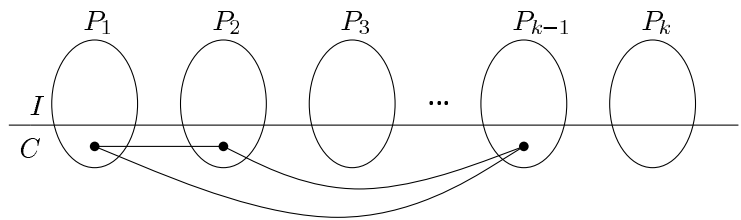}
            \par\end{centering}

            \caption{\label{fig:bipartite_split_graph}Structure of a $k$-partite split
            graph}
            \end{figure}
        \end{proof}

        \begin{prop} \label{prop:large-hom-clique-split-graph}
            Let $B$ be an $\ell\times\ell$ matrix whose diagonal entries are all 1.
            Let $G_{B}$ be a split graph that admits a $B$-partition. Then every
            part $P$ of a $B$-partition of $G_{B}$ contains a homogeneous set in
            $G_{B}$ of size at least $\frac{|P|-1}{2^{\ell-1}}$.
        \end{prop}
        \begin{proof}
            The result follows from Proposition \ref{prop:large-hom-set-split-graph},
            since $G_{B}$ admits a $B$-partition if and only if $\ol{G_{B}}$
            admits a $\ol B$-partition, and the complement of a split graph is
            a split graph.
        \end{proof}

        We also require the following observation.
        \begin{fact} \label{C-has-asterisks-split-graphs}
            Let $M$ be an $(A,B,C)$-block matrix and let $G$ be a split graph. If
            $C$ has an asterisk entry, then $G$ admits an $M$-partition.
        \end{fact}
        \begin{proof}
            If $C$ has an asterisk, then $M$ contains the matrix
            $\left(\begin{smallmatrix}0 & *\\* * & 1 \end{smallmatrix}\right)$ as a
            principal submatrix. Thus $G$ admits this partition by definition of
            split graphs, since every other part may be empty.\qedhere
        \end{proof}

        \begin{proof} [Proof of Theorem~\ref{thm:split-graph-finitely-many-obs}] Let
            $M$ be an $m\times m$ matrix, with $k$ diagonal $0$s and $\ell$ diagonal
            $1$s. Assume without loss of generality that $k\geq\ell$. We show that
            the number of vertices in a split minimal $M$-obstruction is at most

            \[
                2^{k-1}(k+\ell)(2k+3)+1\in O(k^2\cdot2^k)
            \]

            Suppose for contradiction that $G$ is a minimal $M$ obstruction with at
            least $2^{k-1}(k+\ell)(2k+3)+2$ vertices. By Fact
            \ref{C-has-asterisks-split-graphs}, we may assume that the submatrix $C$ has
            no asterisks. Pick an arbitrary vertex $v$ and consider a partition of the
            graph $G-v$ on at least $2^{k-1}(k+\ell)(2k+3)+1$ vertices. As there are
            $k+\ell$ parts in the partition, by the pigeonhole principle there is a
            part, call it $P$, of size at least $2^{k-1}(2k+3)+1$. This part $P$ is
            either an independent set or a clique, and each of these cases will be
            considered separately below. Either way, by Propositions
            \ref{prop:large-hom-set-split-graph} and
            \ref{prop:large-hom-clique-split-graph}, $P$ contains a homogeneous set in
            $A$ or $B$ (depending on whether $P$ is an independent set or a clique) of
            size at least $\frac{|P|-1}{2^{k-1}}\geq2k+3$. Since $C$ has no asterisks,
            this set is homogeneous in $G$. Thus $G-v$ has a homogeneous set of size at
            least $2k+3$, and so $G$ has a homogeneous set $H$ of size at least $k+2$,
            since by the pigeonhole principle at least $k+2$ of the vertices of $P$
            agree on $v$. Now let $w\in H$, consider a partition of $G-w$, and recall
            that $P$ is either an independent set or a clique.
            \begin{caseenv}
                \item [\bf{Case 1.}] 
                    If $P$ is an independent set, then so is $H$; hence, there are at
                    least $k+1$ independent vertices in $G-w$. As there are $\ell\leq k$
                    clique parts in the partition of $G-w$, and no two independent
                    vertices of $H$ may be placed in the same clique part, at least one
                    vertex $w'\in H-\{w\}$ must be placed in an independent part $P'$.
                    Since $w$ is not adjacent to $w'$ and both vertices belong to $H$,
                    $w$ can be added to $P'$, contradicting the minimality of $G$.
                \item [\bf{Case 2.}]
                    If $P$ is a clique then $H-w$ is a clique of size at least $k+1$,
                    and so in the partition of $G-w$, at least one vertex of $H-w$ falls
                    in a clique part $P'$. As in Case $1$, $w$ can be added to $P'$,
                    contradicting minimality.\qedhere
            \end{caseenv}
        \end{proof}

        Since every matrix $M$ has finitely many split minimal obstructions, there is an
        obvious polynomial time algorithm for the $M$-partition problem. However, a 
        more efficient algorithm is described in what follows.  
        
        A matrix $M$ is \emph{crossed} if each non-asterisk entry in its block $C$
        belongs to a row or column in $C$ of non-asterisk entries.  It has been
        shown that if $M$ is a crossed matrix, then the list $M$-partition problem
        for chordal graphs can be solved in polynomial time \cite{Feder2005}. Since
        split graphs are chordal, the same result applies for split graphs, and we
        can use this to solve the $M$-partition problem for split graphs in
        polynomial time, bearing in mind that by Fact
        \ref{C-has-asterisks-split-graphs}, we may assume that the block $C$ has no
        asterisks and so $M$ is crossed.  
        
        \begin{thm} \label{thm:algorithms} If $G$
            is a split graph and $M$ is any matrix, then the $M$-partition problem
            for $G$ can be solved in time $O(n^{k\ell})$.
        \end{thm} 

        When dealing with the $M$-partition problem with lists, it is shown in
        \cite{Feder2005} that there is a matrix $M$ for which the list $M$-partition
        problem is $NP$-complete, even when restricted to chordal graphs. In fact,
        the graphs constructed in that reduction are split graphs so that this list
        $M$-partition problem remains $NP$-complete even for split graphs.

    \section{\label{sec:A-Special-Class}A Special Class of Matrices}

        As seen in Section \ref{sec:Matrix-Partitions-of-split-graphs}, for
        any $m\times m$ matrix $M$, there is an exponential upper bound
        on the size of a largest split minimal $M$-obstruction. In this section
        we show a family of matrices for which this bound is nearly tight.

        For $k,t\in\nats$, with $1\leq t\leq k-1$, let $M_{k,t}$ be a $k\times k$
        matrix with diagonal entries all zero, $t$ ones in row $k$, symmetrically,
        $t$ ones in column $k$ and asterisks everywhere else. By permuting the rows
        and columns of $M_{k,t}$ we assume without loss of generality that the one
        entries of row $k$ are in columns $k-t,...,k-1$ and symmetrically, that the
        one entries of column $k$ are in rows $k-t,...,k-1$.  See Figure
        \ref{fig:M_k_t_examples} for some examples.

        \begin{figure} \centering{\includegraphics{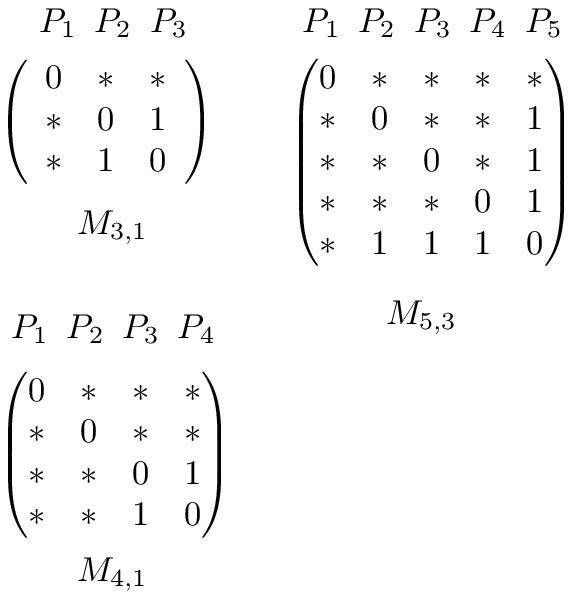}}
            \caption{\label{fig:M_k_t_examples}Matrices $M_{k,t}$ for $k\in\{3,4,5\}$
            and $t\in\{1,3\}$ }
        \end{figure}

        \begin{thm} \label{thm:exponential-lower-bound-thm}
            There exist $k,t\in\nats$ such that for the matrix $M_{k,t}$, the size
            of the largest split minimal $M$-obstruction is $\Omega(2^{k-1})$.
        \end{thm}
        \begin{proof}
            We choose values of $k$ and $t$ so that the matrix $M_{k,t}$ has
            a split minimal obstruction of size at least
            \[
            \left(\pi\cdot\frac{k-1}{2}\right)^{-\frac{1}{2}}\cdot2^{k-1}+2k-1
            \]
            Choose $k=2n+1$ and $t=n$ for some $n\in\nats$, so that the matrix
            $M_{k,t}$ has $2n+1$ parts. Place ones in row $2n+1$ and columns
            $n,n+1,...2n$ as well as in columns $2n+1$ and rows $n,n+1,...,2n$.  Let
            $P$ denote the part in row and column $2n+1$, and designate the $n$
            parts that have a one to $P$ as \emph{restricted parts},
            $R_{1},...,R_{n}$ and the remaining $n$ parts as \emph{unrestricted
                parts}, $U_{1},...,U_{n}$. See Figure \ref{fig:The-matrix_M_n}.

            \begin{figure}[H]
                \begin{centering}
                    \includegraphics{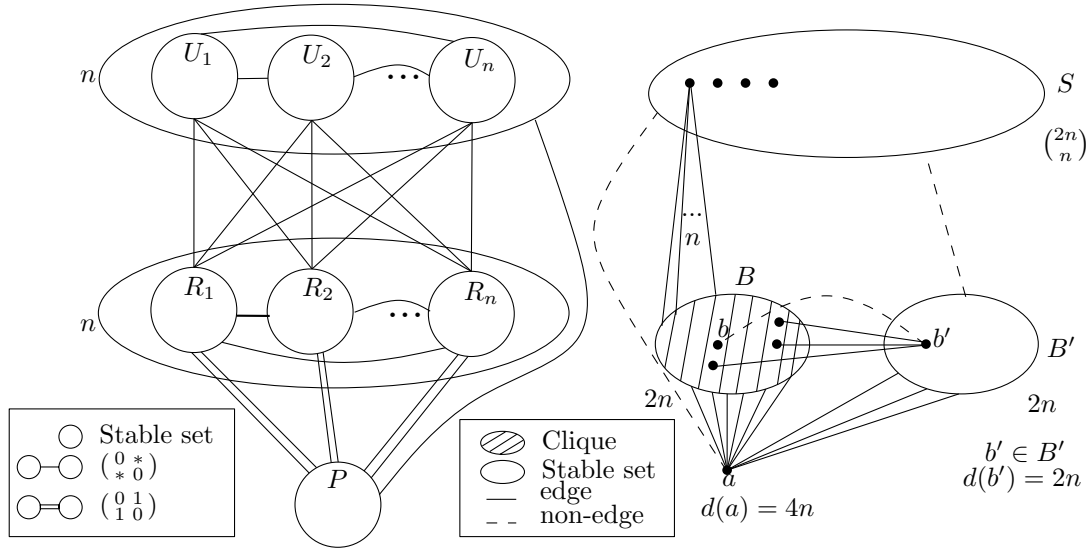}
                    \par
                \end{centering}

                \caption{\label{fig:The-matrix_M_n}The matrix $M_{2n+1,n}$ (left) and an
                obstruction $G$ (right)}
            \end{figure}

            The minimal obstruction $G$, depicted in Figure
            \ref{fig:The-matrix_M_n}, has a \emph{special vertex} $a$, and $2n$
            vertices forming a clique $B$, that are all adjacent to $a$ (so that
            $B\u\{a\}$ is a clique of size $2n+1$). Further, $G$ has another $2n$
            vertices forming an independent set $B'$ such that for each $b\in B$
            there is a $b'\in B'$ that is not adjacent to $b$ but is adjacent to every
            other vertex of $B\u\{a\}$.  Call $b$ and $b'$ \emph{mates}. Finally,
            $G$ has an independent set $S$ of size ${2n \choose n}$ such that for
            every subset $\tilde{B}$ of $B$ of size $n$, there is exactly one vertex
            $s\in S$ adjacent to exactly the vertices of $\tilde{B}$. Note that $G$
            is a split graph since $B\u\{a\}$ is a clique and $B'\u S$ is an
            independent set, as seen in Figure \ref{fig:A-split-partition-of-G}.

            \begin{figure}[H]
            \begin{centering}
            \includegraphics{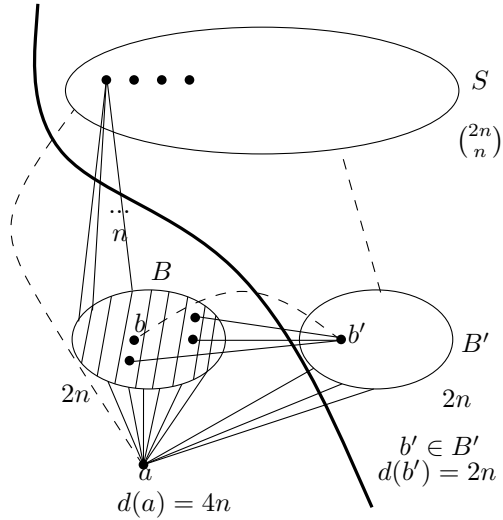}
            \par\end{centering}

            \caption{\label{fig:A-split-partition-of-G}A split partition for $G$.}
            \end{figure}

            To see that $G$ is indeed an obstruction, suppose otherwise, and note
            that $B\u\{a\}$ is a clique of size $2n+1$, so each of its vertices must
            be placed in a different part. Since each vertex of $B$ has a mate in
            $B'$ that is adjacent to $a$ and all of the other vertices in $B$, all
            parts other than the part containing $a$ have size at least two in any
            $M_{k,t}$-partition of $G$. Thus only the part containing $a$ may be a
            singleton. Further $P$ must be the only singleton part, otherwise all of
            the restricted parts must be singletons, since G contains no induced
            $C_4$. Therefore $a\in P$. Now whichever $n$ vertices of $B$ are placed
            in the unrestricted parts, as in
            Figure~\ref{fig:An-attempt-to-partition-G}, there is a vertex $s\in S$
            adjacent to exactly these vertices, and so must be placed into one of
            the restricted parts. But as $s$ is not adjacent to $a$, it cannot be
            placed in a restricted part, and $s$ can't be added to $P$; hence, $G$
            is not $M_{k,t}$-partitionable.

            \begin{figure}[H]
            \begin{centering}
            \includegraphics{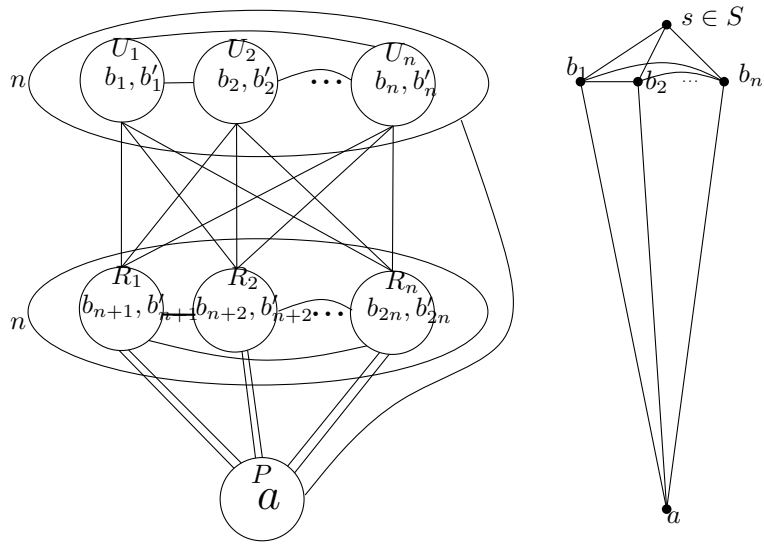}
            \par\end{centering}

            \caption{\label{fig:An-attempt-to-partition-G}An attempt to partition $G$.}
            \end{figure}

            To argue that $G$ is a minimal obstruction, we show that removing
            a vertex from one of $S,B,B',$ or $\{a\}$ allows a partition for
            the resulting graph:
            \begin{enumerate}
                \item [(i)]
                    For $s\in S$ partition $G-s$ as follows: map $a$ to $P$, place
                    each $b\in B$, together with its mate $b\in B'$, in some part,
                    taking care that neighbours of the missing $s$ are placed in
                    unrestricted parts. Now each remaining vertex of $S$ has an
                    unrestricted part to go to.
                \item [(ii)]
                    We consider $b\in B$ together with its mate $b'\in B'$. For
                    $G-b$, place $a$ in $P$, place $b$'s mate $b'$ in an
                    unrestricted part $P_{b'}$, and place all of $S$ and all of $B'$
                    in $P_{b'}$.  This is possible since $B'\u S$ is an
                    independent set. place the remaining $2n-1$ vertices of $B$ in
                    the remaining $2n-1$ parts arbitrarily. To partition $G-b'$,
                    place $b$ in $P$, map $a$ together with all of the
                    vertices of $S$ in an unrestricted part $P_{a}$, and place each
                    other pair of mates $v,v'$ from $B$ and $B'$ into a part,
                    different from $P$ and $P_{a}$.
                \item [(iii)]
                    Finally, $G-a$ can be partitioned using the restricted and
                    unrestricted parts only, not placing anything in $P$. Place each
                    $b$ and its mate $b'$ into a part. Each $s\in S$ is only
                    forbidden from $n$ out of the 2n parts and so can be placed
                    somewhere.
            \end{enumerate}
            Now $G$ has $2k-1+{k-1 \choose \frac{k-1}{2}}=4n+1+{2n \choose n}$
            vertices, and using Stirling's approximation, we get
            \[
            \frac{2^{k-1}}{\sqrt{\pi\frac{k-1}{2}}}=\frac{2^{2n}}{\sqrt{\pi
                    n}}\leq{2n \choose n}\leq\frac{2^{2n}}{\sqrt{\pi
                    n}}\left(1-\frac{c}{n}\right)=\frac{2^{k-1}}{\sqrt{\pi\frac{k-1}{2}}}\left(1-\frac{2c}{k-1}\right),\text{where \ensuremath{\frac{1}{9}<c<\frac{1}{8}}}
            \]
            Therefore $G$ is of size exponential in $k$.
        \end{proof}

    \section{\label{sec:Generalized-Split-Graphs}Generalized Split Graphs}

        Recall that split graphs can be viewed as a special case of $(k, \ell)$-graphs -
        those graphs whose vertices can be partitioned into $k$ independent sets and
        $\ell$ cliques. (Thus split graphs are the $(1, 1)$-graphs.)

        In this section, we focus on $(k, \ell)$-graphs other than the $(1,1)$-graphs.
        We begin with $(2, 0)$- and $(0,2)$-graphs, and then discuss
        other $(k,\ell)$-graphs. Recall that the $(2, 0)$-graphs are the bipartite
        graphs, while the $(0, 2)$-graphs are the co-bipartite graphs. As it turns
        out, there are finitely many bipartite or co-bipartite minimal obstructions,
        for any matrix $M$.

        \begin{thm} \label{thm:bipartite-finitely-many-obs}
            For any $m\times m$ matrix $M$, there are finitely many bipartite
            minimal obstructions and finitely many co-bipartite minimal
            obstructions.
        \end{thm}

        To prove Theorem~\ref{thm:bipartite-finitely-many-obs} we use an approach
        similar in nature to that used
        Section~\ref{sec:Matrix-Partitions-of-split-graphs}. Starting with bipartite
        graphs, note that we may assume that the matrix
        $\left(\begin{smallmatrix}
                    0 & *\\*
                    * & 0
                  \end{smallmatrix}\right)$
        is not a principal submatrix of the matrix $M$, or else the problem would
        be trivial.

        \begin{prop} \label{prop:Large_hom_set_bipartite}
            Let $M$ be an $(A,B,C)$-block matrix, with A of size $k\times k$ and $B$
            of size $\ell\times \ell$. Suppose the block $A$ has no asterisk
            entries. If $G$ is an $M$-partitionable bipartite graph, then any part
            $P$ of $A$ in an $M$-partition of $G$ contains a homogeneous set of size
            at least $\frac{|P|}{2^{2\ell}}$
        \end{prop}
        \begin{proof}
            Fix a bipartition of $G$ and let $P$ be a part of $A$ in an
            $M$-partition of $G$. We argue that $P$ has the desired size. As $A$ has
            no asterisks, the vertices of $P$ all have the same adjacency relation
            to vertices in other parts of $A$. Now let $P'$ be some part of $B$.
            Since $G$ is bipartite, $P'$ can have at most two vertices, one from
            each part of the bipartition of $G$. Let these vertices be $x$ and $y$.
            By the pigeonhole principle, $x$ is either adjacent to, or non-adjacent
            to, at least half of the vertices of $P$. Suppose with out loss of
            generality, that $x$ is adjacent to at least half of the vertices of
            $P$. Call these vertices $P_x$. Applying the pigeonhole principle again,
            this time to the vertex $y$, we have that $y$ is either adjacent to, or
            non-adjacent to, at least half of the vertices of $P_x$. Let the larger
            of these two sets be $P_{xy}$, and note that $P_{xy}\geq\frac{|P|}{2^2}$
            Now there are $\ell-1$ clique parts other than $P'$, each of size at
            most two. Inductively, we obtain a homogeneous set in $P$ of size at
            least $\frac{|P|}{2^{2\ell}}$
        \end{proof}

        Theorem~\ref{thm:bipartite-finitely-many-obs} now follows for bipartite
        graphs. The proof for co-bipartite graphs follows by complementation.

        \begin{proof} [Proof of Theorem~\ref{thm:bipartite-finitely-many-obs}]
            As discussed above, we assume that $A$ contains no asterisk entries. We
            show that any bipartite minimal obstruction is of size at most
            \[
            2^{2\ell}(k+\ell)(2\ell+3)
            \]

            Suppose otherwise, and let $G$ be a minimal obstruction with at least
            $2^{2\ell}(k+\ell)(2\ell+3) + 1$ vertices. For an arbitrary vertex $v$,
            the graph $G-v$ is $M$-partitionable, and so some part $P$ in an
            $M$-partition of $G-v$ contains at least $2^{2\ell}(2\ell+3)$ vertices.
            Since $2^{2\ell}(2\ell+3)\geq 3$ for $l\geq 0$, and no clique part of
            $M$ may contain more than two vertices, $P$ must be an independent set.
            Thus by Proposition~\ref{prop:Large_hom_set_bipartite}, $P$ contains a
            homogeneous set of size at least $\frac{|P|}{2^{2\ell}}\geq 2\ell + 3 $.
            By the pigeonhole principle, $G$ has an homogeneous set $H$ of size at
            least $\ell+2$. Note that $H$ is an independent set. Let $h\in H$, and
            consider a partition of $G-h$. As there are only $\ell$ cliques and
            $\ell+1$ vertices in $H-{h}$, there must be a part $P'$ of $A$ that
            contains a vertex $h'$ of $H-{h}$. But since $H$ is an independent set,
            and $h$ has the same neighbourhood as $h'$, we may add $h$ to $P'$ to
            obtain a partition of $G$, a contradiction.
        \end{proof}

        We now consider $(k,\ell)$-graphs for values of $k$ and $\ell$ that satisfy
        $k+\ell\geq3$. For convenience, let $(k,\ell)$ denote the set of
        $(k,\ell)$-graphs. The family of graphs depicted in
        Figure~\ref{fig:M_3_1_and_obs} is an infinite family of chordal minimal
        obstructions to the matrix $M_{3,1} $\cite{FederHellRizi}. We define the
        family more precisely as follows.

        For $t\geq 3$, let $G(t)$ be the graph consisting of an even path on $2t$
        vertices, and an additional vertex $u$. $u$ is adjacent to all vertices of
        the path, except the endpoints. Note that each $G(t)$ is chordal.

        \begin{thm} \label{thm:generalized-split-inifinitely-many-obs}
            If $k,\ell\in\nats$ such that $k+\ell\geq3$, then there exists a matrix
            $M$ that has infinitely many $(k,\ell)$-minimal obstructions.
        \end{thm}
        \begin{proof}
            Note that for any $t\geq 3$, $G(t)$ is 3-colourable, and $G(t)$ is
            partitionable into a bipartite graph and a clique. That is, $G(t)\in (3,
            0)\i(2,1)$. Therefore, for the matrix $M_{3,1}$, there are infinitely many
            (chordal) minimal $(2,1)\i(3,0)$ obstructions. By complementation, for any
            $t\geq3$, the graph $\overline{G(t)}$ is in $(1,2)\i(0,3)$, providing
            infinitely many (chordal) $(1,2)\i(0,3)$ obstructions for the matrix
            $\overline{M_{3,1}}$.

            Now if $k\leq 1$, then since $k+\ell\geq 3$, it must be that
            $\ell\geq2$, and so the family $\{\overline{G(t)} | t\geq 3\}$ is a family of
            $(k,\ell)$-minimal obstructions for $\overline{M_{3,1}}$. On the other
            hand, if $k\geq2$, then the family $\{G(t)|t\geq 3\}$ is a family of
            $(k,\ell)$-minimal (chordal) obstructions for the matrix $M_{3,1}$.
        \end{proof}

\bibliographystyle{plain}
\bibliography{refs}
\end{document}